\documentclass{article}
\usepackage{latexsym}
\usepackage{verbatim}
\newtheorem{theorem}{Theorem}

\newtheorem{corollary}{Corollary}

\newtheorem{definition}{Definition}

\newcommand{\blackslug}{\mbox{\hskip 1pt \vrule width 4pt height 8pt 
depth 1.5pt \hskip 1pt}}
\newcommand{\qed}{\quad\blackslug\lower 8.5pt\null\par\noindent}
\newenvironment{proof}{\par\noindent{\bf Proof:}}{\qed \par}

\newcommand{\bQ}{\mbox{${\bf Q}$}}

\newcommand{\bH}{\mbox{${\bf H}$}}

\title{Some properties of $n$-party entanglement under LOCC operations
}

\author{Daniel Lehmann\\School of Engineering
and \\ Center for Language, Logic and Cognition
\\Hebrew University, \\Jerusalem 91904, Israel 
}
\date{December 2011}

\begin{document}
\maketitle
\begin{abstract}
Nielsen~\cite{Nielsen:LOCC} characterized in full 
those $2$-party quantum protocols of
local operations and classical communication that transform, with probability one, 
a pure global initial state into a pure global final state. 
The present work considers the generalization
of Nielsen's characterization to $n$-party protocols.
It presents a sweeping generalization of the {\em only if} part of Nielsen's result.
The result presented here pertains also to protocols 
that do not generate a final state for sure,
it considers arbitrary mixed initial states instead of pure states 
and $n$-party protocols for arbitrary $n$'s.
In this very general setting, local operations and classical communication
can never decrease the expected spectra of the local mixed states 
in the majorization ordering. 
In other terms, the local states can only become purer (weakly) in expectation. 
The proof also provides an improvement on Nielsen's.
The {\em if} part of Nielsen's characterization does not generalize. 
This is shown by studying the entanglement of three qubits. 
It is shown that one can find pure states of a system of three 
qubits that are not equivalent under unitary local operations 
but define local mixed states
on all subparts of the system that have the same spectra. 
Neither equivalence of pure states under local unitary operations 
or accessibility under LOCC operations among a system of three
qubits can be characterized by properties of the spectra of the local mixed states. 
\end{abstract}

\section{Introduction and previous work} \label{sec:intro}
We assume each of $n$ parties, i.e., agents, has some piece of a quantic system. 
The pieces do not have to be similar.
Let \mbox{$\bH = \bH_{1} \otimes \bH_{2} \otimes \ldots \otimes \bH_{n}$} 
be a tensor product of $n$ finite-dimensional Hilbert spaces. 
We consider that the global system represented by $\bH$ 
is made of $n$ different parts, 
represented by $\bH_{i}$, for \mbox{$i = 1 , \ldots , n$}, 
the $i$'s part being controlled by agent $i$. 
In accordance with tradition, we assume agent $1$ is Alice.
The different agents may be far away from one another. 
From agent $i$'s point of view, the system it controls is represented 
by a mixed state of $\bH_{i}$
that depends on the global state of the system. 
Different actions of the agents can modify the
global state and therefore the local states. 
This phenomenon is used to realize quantum protocols 
that may achieve a kind of cooperation
between the different agents that is not attainable 
by classical means, as first showed 
by J. S. Bell in~\cite{Bell:64} and described in~\cite{Preskill:notes}, Chapter 4.

We focus here on the case the possible actions of the agents are 
classical communication, local unitary operations 
and local generalized measurements. 
Any protocol of such actions transforms, probabilistically, the initial (mixed) state 
into a final (mixed) state. 
A particularly interesting case of such protocols is that of protocols
in which one final state (pure or mixed) is obtained with probability one.

In~\cite{Nielsen:LOCC}, M. Nielsen characterized 
the final states obtainable from a given
initial {\em pure} state with probability one by local unitary transformations, 
local generalized measurements and classical communication in the $2$-party case 
(\mbox{$n = 2$}). 
He showed that an initial state $\mid \phi \rangle$ can be transformed into 
a final state $\mid \psi \rangle$ in such a way 
iff the spectrum of the local mixed state of Alice 
induced by $\mid \psi \rangle$ majorizes the one induced by $\mid \phi \rangle$.

The $2$-party case has a distinguishing property: 
the spectra of the local mixed states of the two
agents are closely related: they are the same up to, perhaps, some zeros, 
or, in other terms, their strictly positive parts are the same.
This paper generalizes the {\em only if} part of Nielsen's result 
to an $n$-party situation and to initial states that are mixed 
in Section~\ref{sec:Nielsen}. 
Section~\ref{sec:3qu} and later are devoted to showing that the generalization of 
the {\em if} part of Nielsen's result fails quite spectacularly.
Even in an entangled system of three qubits one cannot decide the equivalence 
of pure states under local unitary operations 
by examining only the spectra of the mixed 
states of the different parts of the system.

\section{Positive results} \label{sec:positive}
In this section a sweeping generalization of Nielsen's~\cite{Nielsen:LOCC} result is proposed.
\begin{itemize}
\item Instead of considering only protocols that end up in a final global state for sure, this paper
considers any LOCC protocol and the probability distribution on final global states 
that it generates.
The notion of the expected spectrum of local states is the main tool that will enable us to study
such protocols in general.
\item Instead of considering only pure states as initial and final global states, 
this paper considers probability distributions over mixed states.
\item Instead of considering $2$-party entanglement, this paper considers arbitrary $n$-party
entanglement.
\end{itemize}
\subsection{Local states defined by a global state and their properties} 
\label{sec:mixed-local}
\subsubsection{Local states} \label{sec:local}
Some notation will be useful. 
We shall always use $i$ to represent one of the $n$ agents.
We shall use $i^{-}$ to represent the set of $n - 1$ agents 
that contains all agents except agent $i$. 
In the same spirit, $\bH_{i^{-}}$ represents the tensor product 
of all $\bH_{j}$s, except $\bH_{i}$. 
In other terms \mbox{$\bH_{i^{-}} =$} 
\mbox{$\bH_{1} \otimes \ldots \otimes \bH_{i - 1} \otimes \bH_{i + 1} \otimes \ldots \otimes \bH_{n}$}.
Note that, for every $i$, we have \mbox{$\bH =$} \mbox{$\bH_{i} \otimes \bH_{i^{-}}$}.

Let us assume that the global state is some mixed state, i.e., a linear, self-adjoint, 
weakly positive operator \mbox{$\sigma : \bH \rightarrow \bH$} of trace $1$.
Agent $i$, who sees only the $\bH_{i}$ part of the system, 
sees his system as a mixed state
\mbox{$\rho_{i}^{\sigma} : \bH_{i} \rightarrow \bH_{i}$} defined as:
\begin{equation} \label{eq:mixed-state}
\rho_{i}^{\sigma} \ = \ {Tr}_{\bH_{i^{-}}} (\sigma).
\end{equation}

\subsubsection{Properties of the partial trace operator}
\label{sec:partial-trace}
The partial trace operator satisfies the following properties, for any linear operator
\mbox{$f : A \otimes B \rightarrow A \otimes B$}:
\begin{enumerate}
\item for any basis \mbox{$b_{i}$}, \mbox{$i = 1 , \ldots , n$} of $B$
and for any vectors \mbox{$x , y \in A$},
\begin{equation} \label{eq:basis}
\langle x \mid Tr_{B} ( f ) \mid y \rangle \ = \  
\sum_{i = 1}^{n} \langle x \otimes b_{i} \mid f \mid y \otimes b_{i} \rangle.
\end{equation}
\item the partial trace of the adjoint of an operator is the adjoint of the partial trace:
\begin{equation} \label{eq:trace*}
Tr_{B} ({f}^{\ast}) \ = \ {(Tr_{B}(f))}^{\ast}
\end{equation}
and therefore the partial trace of a self-adjoint operator is self-adjoint,
\item the partial trace of a weakly positive operator is weakly positive,
\item the trace of a partial trace is the trace of the original operator
\begin{equation} \label{eq:trace1}
Tr(Tr_{B}(f)) \ = \ Tr(f)
\end{equation}
\item the partial trace of the identity is multiplication by the dimension of the space, $B$,
on which the partial trace is taken:
\begin{equation} \label{eq:identity}
Tr_{B}( {id}_{A \otimes B} ) \ = \  dim(B) \, {id}_{A},
\end{equation}
\item \label{com} for any linear operator
\mbox{$g : B \rightarrow B$}
\begin{equation} \label{eq:commutation}
Tr_{B} ( f \circ ({id}_{A} \otimes g) ) \ = \  
Tr_{B} (({id}_{A} \otimes g) \circ f),
\end{equation} 
\item for any linear operator
\mbox{$g : A \rightarrow A$}
\begin{equation} \label{eq:self-trace1}
Tr_{B} ( f \circ (g \otimes {id}_{B}) ) \ = \  
Tr_{B} (f) \circ g, 
\end{equation}
and 
\begin{equation} \label{eq:self-trace2}
Tr_{B} ( (g \otimes {id}_{B}) \circ f ) \ = \  
g \circ Tr_{B} (f).
\end{equation}
\end{enumerate}

The case $f$ is the projection $P_{\phi}$ of $A \otimes B$ on a unit vector 
\mbox{$\mid \phi \rangle \in A \otimes B$} is an important special case.
In this case, the global state is a {\em pure} state, i.e., a one-dimensional subspace 
of $A \otimes B$. 

A consequence of the properties of partial trace described above that will be used
in Theorems~\ref{the:unitary-other}, \ref{the:unitary-self} and~\ref{the:equivalence}
will now be proved.
\begin{theorem} \label{the:aux}
Let $\sigma$ be a mixed state of $A \otimes B$ and let
\mbox{$u_{A} : A \rightarrow A$} and \mbox{$u_{B} : B \rightarrow B$} 
be unitary maps. Then,
\[
Tr_{B} ( (u_{A} \otimes u_{B} ) \circ \sigma \circ ( u_{A}^{\ast} \otimes u_{B}^{\ast}) ) 
\ = \ 
u_{A} \circ Tr_{B}( \sigma ) \circ u_{A}^{\ast}.
\]
\end{theorem}
\begin{proof}
\[
Tr_{B} ( (u_{A} \otimes u_{B} ) \circ \sigma \circ ( u_{A}^{\ast} \otimes u_{B}^{\ast}) ) 
\ = \ 
Tr_{B} ( ( u_{A} \otimes u_{B} ) \circ \sigma \circ ( u_{A}^{\ast} \otimes {id}_{B} ) \circ 
( {id}_{A} \otimes u_{B}^{\ast}) ) \ = \ 
\]
by Equation~(\ref{eq:commutation})
\[
Tr_{B} ( ( {id}_{A} \otimes u_{B}^{\ast} ) \circ ( u_{A} \otimes u_{B} ) \circ \sigma \circ 
( u_{A}^{\ast} \otimes {id}_{B} ) ) \ = \ 
Tr_{B} ( ( u_{A} \otimes {id}_{B} ) \circ \sigma \circ 
( u_{A}^{\ast} \otimes {id}_{B} ) ) \ = \ 
\]
by Equations~(\ref{eq:self-trace1}) and~(\ref{eq:self-trace2})
\[
u_{A} \circ Tr_{B}( \sigma ) \circ u_{A}^{\ast}.
\]
\end{proof}

The partial trace operator also satisfies the following.
For any linear operator
\mbox{$f : A \otimes B \otimes C \rightarrow A \otimes B \otimes C$}:
\begin{equation} \label{eq:three}
Tr_{B} ( Tr_{C} ( f ) ) \ = \  Tr_{B \otimes C} ( f ).
\end{equation}

In general, \mbox{$Tr_{B} ( g \circ f ) \neq Tr_{B} ( g ) \circ Tr_{B} ( f )$}.

\subsubsection{Operating on mixed states}
\label{sec:operating}
In Section~\ref{sec:LOCC} we shall discuss some operations 
agents can perform that transform the global state of a system.
We shall now reflect on how transformations of the global state 
should be modeled.
Typically such a transformation is modeled by some linear operator 
\mbox{$ f : \bH \rightarrow \bH$}: any unit vector $x$ of $\bH$ is transformed by, first, 
applying $f$ to it and, then, renormalizing.
If the global state is a mixed state, then, the transformation corresponding to
$f$ is the transformation that transforms the projection on $x$:
\mbox{$P_{x} =$} \mbox{$\mid x \rangle \langle x \mid$} 
into the projection on $f ( x )$:
\mbox{$P_{f (x )} =$}
\mbox{$\mid f x \rangle \langle f x \mid =$} 
\mbox{$ f \mid x \rangle \langle x \mid {f}^{\ast} =$}
\mbox{$f \circ P_{x} \circ {f}^{\ast}$}.
The transform (by $f$) of a mixed state $\rho$ is therefore 
\mbox{$f \circ \rho \circ {f}^{\ast}$} after renormalization:
\begin{equation} \label{eq:rho-prime}
\rho' \ = {{f \circ \rho \circ {f}^{\ast}} \over {Tr(f \circ \rho \circ {f}^{\ast})}}
\end{equation}

One easily sees that, if \mbox{$Tr ( f \circ \rho \circ {f}^{\ast} ) \neq 0$}, $\rho'$ is indeed a self-adjoint, weakly positive operator of trace $1$.
If \mbox{$Tr ( f \circ \rho \circ {f}^{\ast} ) =$} $0$, the state $\rho$ cannot be transformed by $f$.

\subsection{Spectrum and majorization} \label{sec:spectrum}
The spectrum, i.e., the set of eigenvalues (with their multiplicity) of mixed local states 
and the majorization relation between such spectra will prove to be of cardinal importance
in our study of transformations of local states brought about 
by the different agents activities.
Since mixed states are self-adjoint, weakly positive operators of trace one, their spectrum
is composed of real nonnegative numbers whose sum is equal to one. 
{\bf We shall always order such spectra in {\em decreasing} order}.
If \mbox{$Sp(\rho) =$} \mbox{$( \lambda_{1} , \ldots , \lambda_{n} )$}, 
$\lambda_{1}$ is the largest eigenvalue of $\rho$ and $\lambda_{n}$ the smallest.
This enables us to define sums and convex combinations of spectra:
the first element of \mbox{$1 / 2 \: Sp(\rho) \, + \,1 / 2 \: Sp(\sigma)$}, for example, 
is the average of the largest eigenvalues of $\rho$ and $\sigma$.
When we add spectra of different lengths, or compare them as below, 
we always padd the shorter spectrum with zeros on the right.
\begin{definition} \label{def:major}
Let \mbox{$ \lambda =$}
\mbox{$( \lambda_{1} , \ldots , \lambda_{n} )$}, and \mbox{$ \mu =$}
\mbox{$( \mu_{1} , \ldots , \mu_{n} ) $} be spectra (suitably padded).
We say that $\lambda$ {\em majorizes} $\mu$ and write \mbox{$\lambda\succeq \mu$}
iff for every \mbox{$k = 1 , \ldots , n$}, one has:
\begin{equation} \label{eq:major}
\sum_{j = 1}^{k} \lambda_{j} \geq \sum_{j = 1}^{k} \mu_{j}.
\end{equation} 
Let \mbox{$\rho : A \rightarrow A$} and \mbox{$\sigma : B \rightarrow B$} 
be any two mixed states, we say
that $\rho$ {\em majorizes} $\sigma$ and write \mbox{$\rho \succeq \sigma$}
iff \mbox{$Sp(\rho) \succeq Sp(\sigma)$}.
\end{definition}

Note that a pure state majorizes any state. The majorization relation is a pre-order:
reflexive and transitive. Two mixed states are equivalent in the majorization order iff
they have the same spectrum.
The reader should think of Definition~\ref{def:major} in the context of local mixed states
$\rho$ and $\sigma$. 
We now want to define a relation on global states.

\begin{definition} \label{def:global-major}
Let \mbox{$\rho , \sigma : \bH \rightarrow \bH$} be global mixed states, where
\mbox{$\bH =$} \mbox{$\bH_{1} \otimes \ldots \otimes \bH_{n}$}.
We shall say that $\rho$ is {\em stronger} than $\sigma$ iff, 
for every \mbox{$i = 1 , \ldots , n$}, one has:
\begin{equation} \label{eq:purer}
Tr_{\bH_{i^{-}}} ( \rho ) \succeq Tr_{\bH_{i^{-}}} ( \sigma ).
\end{equation}
\end{definition}

The tool we shall use in Section~\ref{sec:local-gen} to prove majorization
properties is Corollary~\ref{co:sum} below. Its proof is based on 
Theorem~\ref{the:Fan} found in Appendix~\ref{sec:Fan} with a proof.

The following is found in~\cite{Marshall-Olkin:Inequalities} p. 241.
\begin{corollary} \label{co:sum}
For any self-adjoint matrices $A$ and $B$, \mbox{$A , B : \bH \rightarrow \bH$}, one has
\mbox{$Sp(A) + Sp(B) \succeq$} \mbox{$Sp ( A + B )$}.
\end{corollary}
\begin{proof}
Let \mbox{$Sp(A + B) =$} \mbox{$\{ \nu_{i} \}_{i = 1 , \ldots , n}$}.
For any $q$, \mbox{$1 \leq q \leq n$}, by Theorem~\ref{the:Fan}, the sum
\mbox{$\sum_{i = 1}^{q} \nu_{i}$} is the maximum value taken by $w_{A + B}(b)$ 
on all $q$-bases for $\bH$. But \mbox{$w_{A + B}(b) =$} \mbox{$w_{A}(b) + w_{B}(b)$},
\mbox{$w_{A}(b) \leq \sum_{i = 1}^{q} \lambda_{i}$} and 
\mbox{$w_{B}(b) \leq \sum_{i = 1}^{q} \mu_{i}$}. 
\end{proof}

\subsection{Local operations and classical communication} \label{sec:LOCC}
\subsubsection{Local operations} \label{sec:local-op}
It is time to have a closer look at those operations that different agents 
that share some entangled system can perform 
in the kind of protocols studied in the theory of quantum information.
Such operations are traditionally described as LOCC: local operations and classical
communication.

From now on, we shall assume that the first agent, the one that sees the $\bH_{1}$
part of the system is called Alice and we shall assume that Alice is the only active agent. 
Obviously, what we say about Alice's actions applies to any other agent's actions.
The characteristic feature of a local operation of Alice is that it acts only on $\bH_{1}$.
Any local operation of Alice can be characterized by a linear operator
\mbox{$f : \bH_{1} \rightarrow \bH_{1}$}. Its effect on the global state $\rho$ is to transform
the global state $\rho$ into \mbox{$f' \circ \rho$} where $f'$ is defined below.
\begin{definition}  \label{def:f'}
For any \mbox{$f : \bH_{1} \rightarrow \bH_{1}$} we let 
\mbox{$f' =$} \mbox{$ f \otimes {id}_{\bH_{1^{-}}} $}.
\end{definition}
There are two kinds of local operations that Alice can perform: 
\begin{itemize}
\item unitary transformations and 
\item measurements.
\end{itemize}
We shall study them in Sections~\ref{sec:unitary} and~\ref{sec:measurements}
respectively. 

\subsubsection{Classical communication} \label{sec:classical}
The second sort of action that Alice can take is to transfer information, in other terms 
{\em talk}, to some of the other agents. 
This is done by classical means and therefore does not 
change the global state. This operation is particularly important in connection with
measurements. Typically, Alice will communicate to her partners the results of 
generalized measurements she has performed. Once this is done, the ensuing operations
of the other agents may depend on the information they received from Alice on the
results of her measurements.

\subsection{Unitary local operations} \label{sec:unitary}
If \mbox{$U : \bH_{1} \rightarrow \bH_{1}$} is unitary, 
then \mbox{$U' =$} \mbox{$U \otimes {id}_{\bH_{i^{-}}} :$}
\mbox{$\bH \rightarrow \bH$} is also unitary.
Therefore \mbox{$Tr ( U' \circ \rho \circ {U'}^{\ast} ) =$}
\mbox{$Tr ( \rho \circ {U'}^{\ast} \circ U' =$}
\mbox{$Tr ( \rho ) =$} $1$ and the normalization factor is $1$. 
The mixed state $\rho$ is
transformed into \mbox{$\rho' =$}
\mbox{$U' \circ \rho \circ {U'}^{\ast}$}, by
Equation~(\ref{eq:rho-prime}).

We can now describe in full the effect of Alice's local unitary operation on all
the local mixed states
\begin{theorem} \label{the:unitary-other}
Alice's local unitary operations do not change the local mixed states of other
agents.
\end{theorem}
\begin{proof}
For \mbox{$i > 1$} 
\begin{equation} \label{eq:unitary-other}
\rho_{i}^{U' \circ \sigma \circ {U'}^{\ast}} \ = \ Tr_{\bH_{i^{-}} } ( ( U \otimes {id}_{\bH_{i^{-}}} ) \circ
\sigma \circ ( {U}^{\ast} \otimes {id}_{\bH_{i^{-}}} ) ) \ = \ 
Tr_{\bH_{i^{-}}} (\sigma) \ = \ \rho_{i}^{\sigma}
\end{equation}
by Theorem~\ref{the:aux}.
\end{proof}
This is indeed as expected: Bob is not affected by and cannot detect a unitary operation of Alice.
The effect of a unitary operation performed by Alice on her own local state is also
as expected.
\begin{theorem} \label{the:unitary-self}
A unitary local operation $U$ of Alice transforms her local mixed state $\rho_{1}^{\sigma}$ into
\mbox{$U \circ \rho_{1}^{\sigma} \circ {U}^{\ast}$}, as if Alice were alone.
Therefore it does not change the spectrum of her local state.
\end{theorem}
\begin{proof}
\begin{equation} \label{eq:unitary-self}
\rho_{1}^{U' \circ \sigma \circ {U'}^{\ast}} \ = \ Tr_{\bH_{1^{-}} } ( ( U \otimes {id}_{\bH_{1^{-}}} ) \circ
\sigma \circ ( {U}^{\ast} \otimes {id}_{\bH_{1^{-}}} ) ) \ = \ 
U \circ \rho_{1}^{\sigma} \circ {U}^{\ast}
\end{equation}
by Theorem~\ref{the:aux}.
Note, now, that for any operator $f$, \mbox{$Sp ( U \circ f \circ {U}^{\ast}) =$}
\mbox{$Sp(f)$}. 
\end{proof}

\subsection{Local measurements} \label{sec:measurements}
\subsubsection{Generalized measurements} \label{sec:generalized}
If one decides to measure some observable represented 
by a generalized measurement, i.e., a sequence 
\mbox{$f_{1} , \ldots , f_{m}$} of operators:
\mbox{$f_{j} : A \rightarrow A$} for \mbox{$j = 1 , \ldots , m$} that satisfy:
\begin{equation} \label{eq:gen}
\sum_{j = 1}^{m} {f_{j}}^{\ast} \circ f_{j} \ = \ {id}_{A}
\end{equation}
one will obtain some result, i.e., some $j$ for his measurement.
The state of the system defines only a probability distribution on the possible results. 
If the state is $\rho$ then the probability of obtaining result $j$ is given by:
\begin{equation} \label{eq:prob-gen}
p_{j} \ = \ Tr ( f_{j}  \circ \rho \circ {f}_{j}^{\ast}) \ = \ Tr ( f_{j} \circ {f}_{j}^{\ast} \circ \rho ).
\end{equation}
The state $\rho$ is changed by the measurement. 
If \mbox{$p_{j} =$} $0$, the result $j$ is never obtained.
If \mbox{$p_{j} >$} $0$, the new state is given by:
\begin{equation} \label{eq:new-state}
\rho'_{j} \ = \ {{1} \over {p_{j}}} \ f_{j} \circ \rho \circ f_{j}^{\ast} .
\end{equation}

A word of caution is in order here. 
The way our generalized measurements have been described above corresponds 
to what is usually named {\em measuring} POVM 
(positive operator valued measure) measurement in the literature,
in which the agent records the result of the measurement performed.
Another form of generalized measurement has also been considered 
in the literature: {\em trace preserving} POVM measurements, in which the agent
does not record the result of the measurement. 
The notion of LOCC operations considered in this paper 
does not include trace preserving operations. 
For a description of these different types of measurement, see, 
for example, the very useful survey~\cite{Plenio-Virmani:intro}.

\subsubsection{Expected spectrum} \label{sec:expected_spectrum}
Measurements, contrary to the unitary and classical operations considered in 
Sections~\ref{sec:unitary} and~\ref{sec:classical} respectively, do not transform a state into
a state. It is a fundamental property of Quantum Physics that measurements transform a state
(pure or mixed) into a probability distribution on states (pure or mixed).
For many purposes, in Quantum Physics, a probability distribution on states can be confused
with the (mixed) state which is the linear combination of the original states weighted by their
respective probabilities. Such a confusion causes no problem as long as the quantic operations
considered are linear. But measurements are not linear transformations and,  
as noticed in Section~\ref{sec:spectrum}, we wish to attach a spectrum with such a probability
distribution. 
The spectrum associated with a system that is in state $\rho_{i}$ with probability
$p_{i}$, \mbox{$i = 1 , 2$}, is \mbox{$\sum_{i = 1}^{2} p_{i} Sp(\rho_{i})$}, i.e.,
the spectrum the $k$'th element of which is a combination of the $k$'th elements of
$\rho_{1}$ and $\rho_{2}$ respectively. This spectrum is not equal to the spectrum of
the state \mbox{$\sum_{i = 1}^{2} p_{i} \rho_{i}$}. In the sequel, probability distributions on
states are not to be confused with (mixed) states.
Section~\ref{sec:local-gen} studies the expected spectrum that results from a local generalized
measurement.
 
\subsubsection{Local generalized measurements} \label{sec:local-gen} 
Suppose Alice performs a local generalized measurement 
\mbox{$f_{1} , \ldots , f_{m} : $} 
\mbox{$\bH_{1} \rightarrow \bH_{1}$} on her piece of the global system, 
which is in global state $\sigma$.
If, following our custom, we let 
\mbox{$f'_{i} : f_{i} \otimes {id}_{\bH_{i^{-}}}$},
the probability of obtaining result $j$ is given by:
\begin{equation} \label{eq:prob-local}
p_{j} \ = \ Tr ( f'_{j} \circ {f'}_{j}^{\ast} \circ \sigma ) \ = \ 
Tr ( Tr_{\bH_{1^{-}}} ( f'_{j} \circ {f'}_{j}^{\ast} \circ \sigma ) ) \ = \ 
\end{equation}
\[
Tr ( f_{j} \circ f_{j}^{\ast} \circ Tr_{\bH_{1^{-}}} ( \sigma) ) \ = \ 
Tr ( f_{j} \circ f_{j}^{\ast} \circ \rho_{1}^{\sigma} )
\]
by Equations~(\ref{eq:prob-gen}), (\ref{eq:trace1}), (\ref{eq:self-trace2})
and~(\ref{eq:mixed-state}). 
We note that, as expected, 
\begin{equation} \label{eq:sum-one}
\sum_{j = 1}^{m} p_{j} \ = \ 
\sum_{j = 1}^{m} Tr ( f_{j} \circ f_{j}^{\ast} \circ \rho_{1}^{\sigma} ) \ = \ 
Tr ( \sum _{j = 1}^{m} f_{j} \circ f_{j}^{\ast}  \circ \rho_{1}^{\sigma} ) \ = \ 
\end{equation}
\[
Tr ( ( \sum _{j = 1}^{m} f_{j} \circ f_{j}^{\ast} ) \, \rho_{1}^{\sigma} ) \ = \ 
Tr ( \rho_{1}^{\sigma} ) \ = \ 1.
\]
Alice may obtain result $j$ for her measurement only if \mbox{$p_{j} > 0$}, and, then, 
the new global state of the system is:
\begin{equation} \label{eq:newsigma}
\sigma'_{j} \ = \ { 1 \over p_{j}} \  f'_{j} \circ \sigma \circ {f'}_{j}^{\ast} .
\end{equation}
If \mbox{$p_{j} =$} $0$, then the operator  
\mbox{$Tr_{\bH_{i^{-}}} ( f'_{j} \circ \sigma \circ {f'}_{j}^{\ast} ) $} is
self-adjoint, weakly positive and has a trace equal to $0$, it is therefore equal to zero.
We conclude that, for any \mbox{$j = 1 , \ldots , m$}, one has:
\begin{equation} \label{eq:allj}
p_{j} \, \sigma'_{j} \ = \ f'_{j} \circ \sigma \circ {f'}_{j}^{\ast}.
\end{equation}

We shall now study the expected effect of a measurement of Alice on each of the local states. 
We shall show that the expected local spectrum always majorizes the current local spectrum.
The proof is based on Corollary~\ref{co:sum}, but 
we need to distinguish two cases. First, we treat the effect of Alice's measurement on other agents' local states, and then its effect on Alice's own local state.
 
\subsubsection{Effect of Alice's measurements on other agents' state} 
\label{sec:otheragent-measurement}

The new local state of any agent $i$, other than Alice, i.e., \mbox{$i > 1$}, 
after Alice has obtained result $j$ is therefore:
\begin{equation} \label{eq:other-than-Alice}
\rho_{i}^{\sigma'_{j}} \ = \ 
Tr_{\bH_{i^{-}}} (\sigma'_{j}) \ = \ 
{1 \over p_{j} } \ Tr_{\bH_{i^{-}}} (f'_{j} \circ \sigma \circ {f'}_{j}^{\ast} )
\end{equation}
by Equation~(\ref{eq:commutation}).

An example of {\em two} entangled systems presented in Appendix~\ref{sec:expect-only}
shows that it is not the case, even for 2-entanglement, that the eigenvalues of
$\rho_{i}^{\sigma'}$ majorize those of $\rho_{i}^{\sigma}$.
An individual measurement by Alice can make the mixed state of 
another agent more chaotic.

Nevertheless, we shall show that, in expectation, taking into account all the
possible results of Alice's measurement, the eigenvalues of the new mixed state
of any agent different from Alice majorize those of the old state of this agent.
The mixed state of the agent becomes less chaotic, {\em purer}. 

\begin{theorem} \label{the:meas-other}
The expected eigenvalues of Bob's local state after a measurement by Alice majorize
the eigenvalues of Bob's current local state.
\begin{equation} \label{eq:meas-other}
\sum_{j = 1}^{m} p_{j} \, Sp( \rho_{i}^{\sigma'_{j}} ) \succeq Sp ( \rho_{i}^{\sigma} )
\end{equation}
for any \mbox{$i > 1$}.
\end{theorem}
\begin{proof}
By Equation~(\ref{eq:other-than-Alice}) we have:
\begin{equation} \label{eq:th1}
\sum_{j = 1}^{m} p_{j} \, Sp( \rho_{i}^{\sigma'_{j}} ) \ = \ 
\sum_{j = 1}^{m} Sp( p_{j} \, \rho_{i}^{\sigma'_{j}} ) \ = \ 
\sum_{j = 1}^{m} Sp ( Tr_{\bH_{i^{-}}} ( {f'}_{j} \circ \sigma \circ {f'}_{j}^{\ast} ) )
\end{equation}
and by Corollary~\ref{co:sum}:
\begin{equation} \label{eq:Fan-other}
\sum_{j = 1}^{m} Sp ( Tr_{\bH_{i^{-}}} ( {f'}_{j} \circ \sigma \circ {f'}_{j}^{\ast} ) ) \succeq
Sp ( \sum_{j = 1}^{m} Tr_{\bH_{i^{-}}} (  {f'}_{j} \circ \sigma \circ {f'}_{j}^{\ast} ) ).
\end{equation}
By property~\ref{com} of the partial trace operator in Section~\ref{sec:mixed-local}:
\[
Tr_{\bH_{i^{-}}} (  {f'}_{j} \circ \sigma \circ {f'}_{j}^{\ast} ) \ = \ 
Tr_{\bH_{i^{-}}} (  {f'}_{j}^{\ast} \circ {f'}_{j} \circ \sigma ) 
\]
and, by Equation~(\ref{eq:gen})
\begin{equation} \label{eq:th2}
\sum_{j = 1}^{m} Tr_{\bH_{i^{-}}} (  {f'}_{j}^{\ast} \circ {f'}_{j} \circ \sigma) ) \ = \ 
Tr_{\bH_{i^{-}}} ( \sum_{j = 1}^{m} {f'}_{j}^{\ast} \circ {f'}_{j} \circ \sigma) \ = \ 
Tr_{\bH_{i^{-}}} ( \sigma ) \ = \ \rho_{i}^{\sigma}.
\end{equation}
\end{proof}

\subsubsection{Effect of Alice's measurement on her own local state} 
\label{sec:self-measurement}
Once Alice has obtained result $j$ her new local state, using Equations~(\ref{eq:self-trace1})
and~(\ref{eq:self-trace2}) is given by:
\begin{equation} \label{eq:Alice-own}
\rho_{1}^{\sigma'_{j}} \ = \ 
Tr_{\bH_{1^{-}}} (\sigma'_{j}) \ = \ 
{1 \over p_{j} } \ Tr_{\bH_{1^{-}}} (f'_{j} \circ \sigma \circ {f'}_{j}^{\ast} ) \ = \ 
\end{equation}
\[
{1 \over p_{j} } \ f_{j} \circ Tr_{\bH_{1^{-}}} ( \sigma ) \circ f_{j}^{\ast} \ = \
{1 \over p_{j} } \ f_{j} \circ \rho_{1}^{\sigma} \circ f_{j}^{\ast}.
\]
Indeed, Alice's state can be computed locally, using the local mixed state $\rho_{1}^{\sigma}$
and the local operations $f_{j}$.
We may now prove a result similar to Theorem~\ref{the:meas-other}.
\begin{theorem} \label{the:meas-self}
The expected eigenvalues of Alice's local state after her measurement majorize
the eigenvalues of her current local state.
\begin{equation} \label{eq:meas-Alice}
\sum_{j = 1}^{m} p_{j} \, Sp( \rho_{1}^{\sigma'_{j}} ) \succeq Sp ( \rho_{1}^{\sigma} ).
\end{equation}
\end{theorem}
\begin{proof}
By Equation~(\ref{eq:Alice-own})
\begin{equation} \label{eq:th3}
\sum_{j = 1}^{m} p_{j} \, Sp( \rho_{1}^{\sigma'_{j}} ) \ = \ 
\sum_{j = 1}^{m} Sp( p_{j} \, \rho_{1}^{\sigma'_{j}} ) \ = \ 
\sum_{j = 1}^{m} Sp ( f_{j} \circ \rho_{1}^{\sigma} \circ f_{j}^{\ast} )
\end{equation}
The operator $\rho_{1}^{\sigma}$ is self-adjoint and weakly positive, it has therefore
a square root, i.e., a self-adjoint, weakly positive operator \mbox{$\alpha : \bH_{1} \rightarrow \bH_{1}$}
such that \mbox{$\rho_{1}^{\sigma} =$} \mbox{$\alpha \circ {\alpha}^{\ast}$}. 
Let \mbox{$\beta_{j} =$} \mbox{$f_{j} \circ \alpha$}. We have
\mbox{$f_{j} \circ \rho_{1}^{\sigma} \circ f_{j}^{\ast} =$} \mbox{$\beta_{j} \circ {\beta}_{j}^{\ast}$}. 
But \mbox{$Sp ( \beta_{j} \circ {\beta}_{j}^{\ast} ) =$} \mbox{$Sp ( {\beta}_{j}^{\ast} \circ \beta_{j} )$},
as is proved in Theorem~\ref{the:sp=} in Appendix~\ref{sec:well-known}.
We conclude that
\[
Sp ( f_{j} \circ \rho_{1}^{\sigma} \circ f_{j}^{\ast} ) \ = \ 
Sp ( {\alpha}^{\ast} \circ f_{j}^{\ast} \circ f_{j} \circ \alpha ).
\]
By Corollary~\ref{co:sum} and Equation~(\ref{eq:gen}) one has:
\begin{equation} \label{eq:Fan-Alice}
\sum_{j = 1}^{m} Sp ( \alpha^{\ast} \circ f_{j}^{\ast} \circ f_{j} \circ \alpha ) \succeq
Sp ( \sum_{j = 1}^{m}  \alpha^{\ast} \circ f_{j}^{\ast} \circ f_{j} \circ \alpha) \ = \ 
Sp ( \alpha^{\ast} \circ \alpha ) \ = \ Sp ( \rho_{1}^{\sigma} ).
\end{equation}
\end{proof}

Theorem~\ref{the:meas-self} also holds when Alice is alone in the universe:
in expectation the result of a generalized measurement always majorizes the 
initial state. As a consequence, in expectation, the entropy cannot increase as 
a result of a measurement. This fits in well with the idea that a measurement 
always reduces uncertainty.

\subsection{LOCC operations weakly increase the spectra of all local states in the majorization order}
\label{sec:increase}
\begin{theorem} \label{the:increase}
In any LOCC protocol, the spectrum of any initial local state is majorized by its expected final 
local spectrum in the majorization order.
\end{theorem}
\begin{proof}
By Section~\ref{sec:classical}, Theorems~\ref{the:unitary-other}, \ref{the:unitary-self},
\ref{the:meas-other} and~\ref{the:meas-self} no step in the protocol can decrease any 
local spectrum in the majorization order.
\end{proof}

\subsection{Derivation of a generalization of one-half of Nielsen's result} \label{sec:Nielsen}
We can now derive a generalization of one half (the only if part) of Nielsen's
Theorem 1 in~\cite{Nielsen:LOCC}.

\begin{corollary} \label{the:onlyif-Nielsen}
If there is an $n$-party protocol consisting of local unitary operations, 
local generalized measurements and classical communication that, starting in a 
mixed global state $\sigma$ terminates {\em for sure}, i.e., with probability one, in mixed 
global state $\sigma'$, then $\sigma'$ is stronger than $\sigma$ in the sense of
Definition~\ref{def:global-major}.
\end{corollary}
\begin{proof}
At each step of the protocol, we have shown that, for any agent, the initial mixed local state is 
majorized by the expected spectrum of the final mixed local state. If the final global state
is, for sure, $\sigma'$, the final mixed local states are $\rho_{i}^{\sigma'}$ and the expected
spectra are $Sp ( \rho_{i}^{\sigma'} )$. We conclude that, for every $i$, \mbox{$1 \leq i \leq n$}
one has: \mbox{$\rho_{i}^{\sigma'} \succeq \rho_{i}^{\sigma}$}, proving our claim.
\end{proof}

One may note that our results do not use Schmidt's decomposition, which is used heavily
in~\cite{Nielsen:LOCC}.

\subsection{Generalization to an arbitrary probability $p$} \label{sec:p}
We generalize the {\em only if} part of Nielsen's result (Corollary~\ref{the:onlyif-Nielsen}) to
the case the LOCC protocol ends in state $\sigma'$ only with a probability $p$ that may be 
less than unity.

First, we define a relation of approximate majorization that generalizes Definitions~\ref{def:major}
and~\ref{def:global-major}.
\begin{definition} \label{def:approximate_major}
Let $c$ be a real number. Let \mbox{$ \lambda =$}
\mbox{$( \lambda_{1} , \ldots , \lambda_{n} )$}, and \mbox{$ \mu =$}
\mbox{$( \mu_{1} , \ldots , \mu_{n} ) $} be spectra (suitably padded).
We say that $\lambda$ {\em $c$-majorizes} $\mu$ and write \mbox{$\lambda\succeq_{c} \mu$}
iff for every \mbox{$k = 1 , \ldots , n$}, one has:
\begin{equation} \label{eq:approximate_major}
\sum_{j = 1}^{k} \lambda_{j} \geq \sum_{j = 1}^{k} \mu_{j} - 1 + c.
\end{equation} 
Let \mbox{$\rho : A \rightarrow A$} and \mbox{$\sigma : B \rightarrow B$} 
be any two mixed states, we say
that $\rho$ {\em $c$-majorizes} $\sigma$ and write \mbox{$\rho \succeq_{c} \sigma$}
iff \mbox{$Sp(\rho) \succeq_{c} Sp(\sigma)$}.
If $\rho$ and $\sigma$ are global mixed states, as in Definition~\ref{def:global-major}, we say
that $\rho$ is $c$-stronger than $\sigma$ iff, 
for every \mbox{$i = 1 , \ldots , n$}, one has:
\begin{equation} \label{eq:approximate_purer}
Tr_{\bH_{i^{-}}} ( \rho ) \succeq_{c} Tr_{\bH_{i^{-}}} ( \sigma ).
\end{equation}
\end{definition}

The relation $\succeq_{c}$ is not, in general, transitive but it obviously 
satisfies the following.
\begin{theorem} \label{the:succ-c}
\begin{enumerate}
\item For any real numbers $c$ and $d$, such that \mbox{$c > d$}, 
for any $\sigma$ and $\tau$, if \mbox{$\sigma \succeq_{c} \tau$}, then
\mbox{$\sigma \succeq_{d} \tau$},
\item for any \mbox{$c > 1$}, for no $\sigma$ and $\tau$ do we have
\mbox{$\sigma \succeq_{c} \tau$},
\item the relation $\succeq_{1}$ is the majorization relation $\succeq$,
\item for any \mbox{$c \leq 0$}, for any $\sigma$, $\tau$ we have
\mbox{$\sigma \succeq_{c} \tau$}.
\end{enumerate}
\end{theorem}

\begin{theorem}
If there is an $n$-party protocol consisting of local unitary operations, 
local generalized measurements and classical communication that, starting in a 
mixed global state $\sigma$ terminates {\em with probability at least} $p$, in mixed 
global state $\sigma'$, then $\sigma'$ is $p$-stronger than $\sigma$.
\end{theorem}
\begin{proof}
Suppose \mbox{$q \geq p$} is the probability with which state $\sigma'$ is attained.
By Theorem~\ref{the:increase} we have
\mbox{$q \, Sp(Tr_{i^{-}}(\sigma')) + (1 - q) \, S \succeq$} \mbox{$Sp(Tr_{i^{-}}(\sigma))$}, where 
$S$ is the spectrum expected if the protocol does not attain $\sigma'$.
But \mbox{$q \leq 1$} and \mbox{$1 - q \leq 1 - p$} and we have:
\mbox{$Sp(Tr_{i^{-}}(\sigma')) + (1 - p) \, S \succeq$} \mbox{$Sp(Tr_{i^{-}})(\sigma)$}.
Denote the eigenvalues of $Tr_{i^{-}}(\sigma)$ and $Tr_{i^{-}}(\sigma')$ by $\lambda_{i}$ and
$\lambda'_{i}$ respectively and by $\mu_{i}$ those of $S$.
For any \mbox{$k \leq n$}, we have:
\mbox{$\sum_{i = 1}^{k} \lambda'_{i} + (1 - p) \, \sum_{i = 1}^{k} \mu_{i} \geq$}
\mbox{$\sum_{i = 1}^{k} \lambda_{i}$} and therefore
\mbox{$\sum_{i = 1}^{k} \lambda'_{i} + 1 - p \geq$}
\mbox{$\sum_{i = 1}^{k} \lambda_{i}$}.
\end{proof}
 
\section{Negative results} \label{sec:negative}
In this section, we describe a number of ways certain generalizations 
of {\em if} part of Nielsen's result fail.

\subsection{Mixed states} \label{sec:partial-mixed}
The converse of Corollary~\ref{the:onlyif-Nielsen} does not hold. 
In fact, it fails quite dramatically.
Consider any mixed global state $\sigma$ and the mixed global state 
\mbox{$\sigma' =$}
\mbox{$\rho_{1}^{\sigma} \otimes \ldots \otimes \rho_{n}^{\sigma}$}.
It is clear that, for every $i$, \mbox{$1 \leq i \leq n$}, 
\mbox{$\rho_{i}^{\sigma'} =$} \mbox{$\rho_{i}^{\sigma} $} 
and therefore $\sigma'$ and $\sigma$ are equivalently strong. 
But, there is a protocol transforming, for sure, 
$\sigma'$ into $\sigma$ only if $\sigma$ is itself a product state 
since product states, such as $\sigma'$, 
are transformed into product states by local operations. 
In general, therefore, there is no protocol transforming $\sigma'$ into $\sigma$.

Any LOCC protocol transforms a pure state, 
or a distribution over such states, into a distribution
over pure states.
One possible generalization of Nielsen's result may be to ask whether, 
given two pure global states $s$ and $s'$ such that $s'$ is
stronger than $s$ there is always a transformation of $s$ into $s'$, for sure, 
by local operations and classical communication. 
Nielsen has shown that, for $2$-entanglement, this is the case.
Theorem~\ref{the:3qubits-negative} will show, 
by studying systems of three qubits,
that, for $n$-entanglement with \mbox{$n > 2$}, this is not the case.

Another angle of attack may be to try and generalize Nielsen's result, 
for \mbox{$n = 2$} to mixed states. 
A serious problem is already apparent when one studies equivalence 
of mixed global states under local unitary operations for $2$-entanglement. 
Let $\sigma$, $\tau$ be mixed states of $A \otimes B$.
If there are unitary local operations \mbox{$u_{A} : A \rightarrow A$} and
\mbox{$u_{B} : B \rightarrow B$} such that 
\mbox{$\tau =$} 
\mbox{$ (u_{A} \otimes u_{B} ) \circ \sigma \circ (u_{A}^{\ast} \otimes u_{B}^{\ast} )$}
then, necessarily, the states $\sigma$ and $\tau$ have the same spectrum,
the local states \mbox{$Tr_{B} ( \sigma )$} and \mbox{$Tr_{B} ( \tau )$} 
have the same spectrum and so do \mbox{$Tr_{A} ( \sigma )$} 
and \mbox{$Tr_{A} ( \tau )$}. 
Theorem~\ref{the:no-mixed-2} will show that, even for two qubits, 
there are such states $\sigma$ and $\tau$ that are not equivalent 
under local unitary operations.
 
 \subsection{Three qubits} \label{sec:3qu}
To every unit vector $x$ of $A$ one associates its projection $p_{x}$,  denoted
\mbox{$\mid x \rangle \langle x \mid$} in Dirac's notation, which is a mixed state of $A$.
A qubit \bQ\ is a two-dimensional Hilbert space on the complex field.
Let  \mbox{$\bH = \bQ_{1} \otimes \bQ_{2} \otimes \bQ_{3}$} be the tensor product of three qubits. 
Given any unit vector \mbox{$x \in \bH$} one defines mixed states on each of the qubits by:
\begin{equation} \label{eq:mixed}
\rho_{i}^{x} \ = \ Tr_{\bH_{i^{-}}} ( p_{x} )
\end{equation}
for \mbox{$i = 1 , 2 , 3$}, where \mbox{$\bH_{1^{-}} =$} \mbox{$\bH_{2} \otimes \bH_{3}$},
\mbox{$\bH_{2^{-}} =$} \mbox{$\bH_{1} \otimes \bH_{3}$} and 
\mbox{$\bH_{3^{-}} =$} \mbox{$\bH_{1} \otimes \bH_{2}$}.
Let \mbox{$\mid 0_{i} \rangle$} and \mbox{$\mid 1_{i} \rangle$} 
be a basis for $\bQ_{i}$,
for \mbox{$i = 1 , 2 , 3$}. 
When the sub-index is obvious from the context we shall not mention it 
and we shall abuse notations. 
For example \mbox{$\mid 0 1 0 \rangle$} denotes
the product state of $\bH$: 
\mbox{$\mid 0_{1} \rangle \otimes \mid 1_{2} \rangle \otimes
\mid 0_{3} \rangle$}.
In the sequel, indices $i$ and $j$ will range over the qubits:
\mbox{$i , j \in \{ 1 , 2 , 3 \}$} and $k$, $l$ and $m$ range over the indices 
of the base vectors: \mbox{$k , l , m \in \{ 0 , 1 \}$}.
In formulas using summation over those indices 
we shall dispense with specifying the bounds.
 
\subsection{The question} \label{sec:question}
The main technical question we ask and answer is the following: 
given any three mixed qubit states:
\mbox{$\sigma_{i} : \bQ_{i} \rightarrow \bQ_{i}$}, 
is there some {\em pure} state $x$ 
of \bH\ such that \mbox{$\rho_{i}^{x} =$} \mbox{$\sigma_{i}$} for every $i$? 
Our answer to the question is not complete, but we shall learn enough to
demonstrate that $n$-party entanglement 
for \mbox{$n > 2$} has properties very different from $2$-party entanglement.

\subsection{The equations} \label{sec:equations}
Without loss of generality, we let, for any $i$, the eigenvalues of $\sigma_{i}$ be
\mbox{$\lambda_{i}^{0} \geq \lambda_{i}^{1} \geq 0$} 
such that \mbox{$\sum_{k} \lambda_{i}^{k} =$} $1$,
we assume, w.l.o.g., that \mbox{$\lambda_{1}^{0} \geq \lambda_{2}^{0} \geq \lambda_{3}^{0}$} 
and we let $\mid 0_{i} \rangle$
be an eigenvector of $\sigma_{i}$ for the eigenvalue $\lambda_{i}^{0}$ 
and $\mid 1_{i} \rangle$
be an eigenvector of $\sigma_{i}$ for the eigenvalue $\lambda_{i}^{1}$.

Let \mbox{$x = \sum_{k , l , m} x_{k l m} \mid k l m \rangle$} be any vector of \bH.
We are now going to write down equations (in the complex coefficients $x_{k l m}$) 
to characterize those pure states $x$ such that \mbox{$\rho_{i}^{x} =$} 
\mbox{$\sigma_{i}$}, for any $i$.

A first equation requires $x$ to be a unit vector:
\begin{equation} \label{eq:unit}
1 \ = \ || x || \ = \ \sum_{k , l , m} x_{k l m} x_{k l m}^{\ast}.
\end{equation}

Our next equations express that \mbox{$\mid 0_{i} \rangle$} is an eigenvector
of $\sigma_{i}$ for eigenvalue $\lambda_{i}^{0}$.
We have:
\begin{equation} \label{eq:zero-eigen1}
\lambda_{i}^{0} \ = \ \langle 0_{i} \mid \sigma_{i} \mid 0_{i} \rangle \ = \ 
\sum_{k , l} {x}_{\tilde{i}(0 k l)} {x}_{\tilde{i}(0 k l)}^{\ast}
\end{equation}
where \mbox{$\tilde{1}(k l m) =$} \mbox{$k l m$}, 
\mbox{$\tilde{2}(k l m) =$} \mbox{$m k l$}, 
\mbox{$\tilde{3}(k l m) =$} \mbox{$l m k$} and
\begin{equation} \label{eq:zero-eigen2}
0 \ = \ \langle 1_{i} \mid \sigma_{i} \mid 0_{i} \rangle \ = \ 
\sum_{k , l} {x}_{\tilde{i}(1 k l)} {x}_{\tilde{i}(0 k l)}^{\ast}.
\end{equation}
Our last equations express  that \mbox{$\mid 1_{i} \rangle$} is an eigenvector
of $\rho_{i}$ for eigenvalue $\lambda_{i}^{1}$.
\begin{equation} \label{eq:one-eigen1}
\lambda_{i}^{1} \ = \ \langle 1_{i} \mid \sigma_{i} \mid 1_{i} \rangle \ = \ 
\sum_{k , l} {x}_{\tilde{i}(1 k l)} {x}_{\tilde{i}(1 k l)}^{\ast}
\end{equation}
and
\begin{equation} \label{eq:one-eigen2}
0 \ = \ \langle 0_{i} \mid \sigma_{i} \mid 1_{i} \rangle \ = \ 
\sum_{k , l} {x}_{\tilde{i}(0 k l)} {x}_{\tilde{i}(1 k l)}^{\ast}.
\end{equation}

One notices that, for every $i$, since \mbox{$\sum_{k} \lambda_{i}^{k} =$} $1$, one can obtain
Equation~(\ref{eq:one-eigen1}) by subtracting Equation~(\ref{eq:zero-eigen1}) from
Equation~(\ref{eq:unit}) and that Equation~(\ref{eq:one-eigen2}) is implied, by transposition, 
by Equation~(\ref{eq:zero-eigen2}) .
Therefore we conclude:
\begin{theorem} \label{the:char1}
The 7 Equations~(\ref{eq:unit}), (\ref{eq:zero-eigen1}) and (\ref{eq:zero-eigen2}) 
characterize those vectors $x$ for which  \mbox{$\rho_{i}^{x} =$} \mbox{$\sigma_{i}$} 
for every $i$.
\end{theorem}

\subsection{Study of solutions} \label{sec:solutions}
We want to study solutions to those equations, but are far from a complete understanding.
We shall see in Theorem~\ref{the:nosol} that the equations above do not always, i.e., for any $\lambda_{i}$'s, 
have a solution, but we shall describe, in Theorem~\ref{the:existence}, a 3-dimensional 
domain (in the $\lambda$'s) in which solutions always exist and a sub-domain in which 
multiple solutions coexist.

Our first result is a full characterization of the solutions in the special case in which the
system is a tensor product of a first qubit and a sytem of two qubits, i.e., in the case
\mbox{$\lambda_{1}^{0} =$} $1$.
\begin{theorem} \label{the:nosol}
For \mbox{$\lambda_{1}^{0} =$} $1$, 
Equations~(\ref{eq:unit}), (\ref{eq:zero-eigen1}) and (\ref{eq:zero-eigen2}) 
have a solution iff \mbox{$\lambda_{2}^{0} =$} \mbox{$\lambda_{3}^{0}$},
equivalently \mbox{$\lambda_{1}^{0} + \lambda_{2}^{0} - \lambda_{3}^{0} \leq$} $1$.
\end{theorem}
\begin{proof}
Two proofs will be presented.
\begin{enumerate}
\item First proof: if \mbox{$\lambda_{1}^{0} =$} $1$, 
the global state $x$ is a tensor product \mbox{$w \otimes y$} where
\mbox{$w \in \bQ_{1}$} and \mbox{$y \in \bQ_{2} \otimes \bQ_{3}$}.
Schmidt's decomposition of $y$ shows that the spectra of the mixed local states
defined by $y$ are equal, i.e., \mbox{$\lambda_{2}^{0} =$} \mbox{$\lambda_{3}^{0}$}.
Conversely, if \mbox{$\lambda_{2}^{0} =$} \mbox{$\lambda_{3}^{0}$} the eigenvectors of the
two mixed local states define the Schmidt's decomposition of a suitable $y$.
\item Second proof:
Suppose $x_{k l m}$ is a solution.
Since \mbox{$\lambda_{1}^{0} =$} $1$, Equations~(\ref{eq:unit}) and~(\ref{eq:zero-eigen1})
for \mbox{$i = 1$} imply that \mbox{$x_{1 l m} = 0$} for any \mbox{$l , m$}.
Equations~(\ref{eq:zero-eigen2}) for \mbox{$i = 2 , 3$} can now be written as:
\begin{equation} \label{eq:matrix}
\left[ \begin{array}{c} x_{0 1 0} \ \ x_{0 0 1}^{\ast} \\ x_{0 0 1} \ \ x_{0 1 0}^{\ast} \end{array} \right] 
\left[ \begin{array}{c} x_{0 0 0}^{\ast} \\  x_{0 1 1} \end{array} \right] \ = \ 0.
\end{equation}
We conclude that:
\begin{itemize}
\item either \mbox{$x_{0 0 0} =$} \mbox{$x_{0 1 1} =$} $0$, in which case 
Equations~(\ref{eq:zero-eigen1}) can be written
\[
1 \ = \ \lambda_{1}^{0} \ = \ x_{0 0 1} x_{0 0 1}^{\ast} + x_{0 1 0} x_{0 1 0}^{\ast} \ , \ 
\lambda_{2}^{0} \ = \ x_{0 0 1} x_{0 0 1}^{\ast} \ , \ \lambda_{3}^{0} \ = \ x_{0 1 0} x_{0 1 0}^{\ast}
\]
and we conclude that \mbox{$\lambda_{2}^{0} + \lambda_{3}^{0} = 1$} and therefore
\mbox{$\lambda_{2}^{0} =$} \mbox{$\lambda_{3}^{0} =$} \mbox{$ 1 / 2$},
\item or \mbox{$x_{0 1 0} x_{0 1 0}^{\ast} =$} \mbox{$x_{0 0 1} x_{0 0 1}^{\ast}$}, in which case 
Equations~(\ref{eq:zero-eigen1}) for \mbox{$i = 2 , 3$} now imply 
\mbox{$\lambda_{2}^{0} =$} \mbox{$\lambda_{3}^{0}$}.
\end{itemize}

For the {\em if} part, notice that if \mbox{$\lambda_{1}^{0} =$} $1$, 
\mbox{$\lambda_{2}^{0} =$} \mbox{$\lambda_{3}^{0}$},
\mbox{$x_{1 l m} =$} $0$ for every $l , m$ and \mbox{$x_{0 10} =$}
\mbox{$x_{0 0 1} =$} $0$, then the equations boil down to:
\begin{equation} \label{eq:boil}
1 = x_{0 0 0} x_{0 0 0}^{\ast} + x_{0 1 1} x_{0 1 1}^{\ast} \ , \ 
\lambda_{2}^{0} = x_{0 0 0} x_{0 0 0}^{\ast}
\end{equation}
which can be solved.
\end{enumerate}
\end{proof}

\begin{definition} \label{def:odd-even}
We shall say that an index $k l m$ is {\em odd} iff the number of ones is odd,
and that it is {\em even} iff the number of ones is even.
\end{definition}

Two families of solutions will be presented, each under a condition concerning the 
$\lambda$'s. Note that the second condition 
\mbox{$\lambda_{1}^{0} + \lambda_{2}^{0} + \lambda_{3}^{0} \leq 2$} implies 
the first condition  \mbox{$\lambda_{1}^{0} + \lambda_{2}^{0} - \lambda_{3}^{0} \leq 1$},
which holds on a larger domain of parameters.

\begin{theorem} \label{the:existence}
Solutions to Equations~(\ref{eq:unit}), (\ref{eq:zero-eigen1}) and~(\ref{eq:zero-eigen2}) for
\mbox{$i = 1 , 2 , 3$} are provided, 
for any phases \mbox{$\theta_{k l m} \in [ 0 , 2 \pi ]$}, by:
\begin{itemize} \item 
if \mbox{$\lambda_{1}^{0} + \lambda_{2}^{0} - \lambda_{3}^{0} \leq 1$}
\begin{equation} \label{eq:leq1}
y_{k l m} \ = \ 0 {\rm \ for \  every \ odd \ index\ } k l m 
\end{equation}
\begin{equation} \label{eq:leq2}
y_{000} \ = \ {e}^{\theta_{000}} \, \sqrt{2 (\lambda_{1}^{0} + 
\lambda_{2}^{0} + \lambda_{3}^{0} - 1)} \: / \: 2 
\end{equation}
\begin{equation} \label{eq:leq2bis}
y_{011} \ = \ {e}^{\theta_{011}} \, \sqrt{2 ( \lambda_{1}^{0} - 
\lambda_{2}^{0} - \lambda_{3}^{0} + 1)} \: / \: 2  
\end{equation}
\begin{equation} \label{eq:leq2ter}
y_{1 0 1} \ = \ {e}^{\theta_{101}} \, \sqrt{2 (- \lambda_{1}^{0} + 
\lambda_{2}^{0} - \lambda_{3}^{0} + 1)} \: / \: 2 
\end{equation}
\begin{equation} \label{eq:leq2quad}
y_{1 1 0} \ = \ {e}^{\theta_{110}} \, \sqrt{2 (- \lambda_{1}^{0} - 
\lambda_{2}^{0} + \lambda_{3}^{0} + 1)} \: / \: 2
\end{equation}
\item 
if \mbox{$\lambda_{1}^{0} + \lambda_{2}^{0} + \lambda_{3}^{0} \leq 2$}
\begin{equation} \label{eq:geq1}
z_{k l m} \ = \ 0 {\rm \ for \  every \ even \ index\ } k l m 
\end{equation}
\begin{equation} \label{eq:geq2}
z_{001} \ = \ {e}^{\theta_{001}} \, \sqrt{2 (\lambda_{1}^{0} + 
\lambda_{2}^{0} - \lambda_{3}^{0})} \: / \: 2 
\end{equation}
\begin{equation} \label{eq:geq2bis}
z_{010} \ = \ {e}^{\theta_{010}} \, \sqrt{2 ( \lambda_{1}^{0} - 
\lambda_{2}^{0} + \lambda_{3}^{0} )} \: / \: 2  
\end{equation}
\begin{equation} \label{eq:geq2ter}
z_{1 0 0} \ = \ {e}^{\theta_{100}} \, \sqrt{2 (- \lambda_{1}^{0} + 
\lambda_{2}^{0} + \lambda_{3}^{0})} \: / \: 2 
\end{equation}
\begin{equation} \label{eq:geq2quad}
z_{1 1 1} \ = \ {e}^{\theta_{111}} \, \sqrt{2 (- \lambda_{1}^{0} - 
\lambda_{2}^{0} - \lambda_{3}^{0} + 2)} \: / \: 2
\end{equation}

\end{itemize}
\end{theorem}
Notice that the quantities of which we take a square root are indeed non-negative:
since
\mbox{$1 \geq \lambda_{1}^{0} \geq \lambda_{2}^{0} \geq \lambda_{3}^{0} \geq 1 / 2$} 
(Equations~(\ref{eq:geq2}) , (\ref{eq:geq2bis}), (\ref{eq:geq2ter}), (\ref{eq:leq2}), (\ref{eq:leq2bis}) 
and~(\ref{eq:leq2ter})
and by the specific assumptions (Equations~(\ref{eq:geq2quad}) and~(\ref{eq:leq2quad})).
\begin{proof}
By inspection. Notice that, in Equations~(\ref{eq:zero-eigen2}), every term is the product
of a variable of odd index by a variable of even index and therefore each of 
Equations~(\ref{eq:geq1}) and~(\ref{eq:leq1}) guarantees that they are satisfied.
\end{proof}

Note that, if \mbox{$\lambda_{1}^{0} =$} $1$, as in Theorem~\ref{the:nosol}, the first
condition \mbox{$\lambda_{1}^{0} + \lambda_{2}^{0} + \lambda_{3}^{0} \leq$} $2$ is equivalent 
to \mbox{$\lambda_{2}^{0} =$} \mbox{$\lambda_{3}^{0} =$} $1 / 2$ and that, in this case, the 
solution $z$ of Theorem~\ref{the:existence} is different from the solution $x$ of 
Theorem~\ref{the:nosol}. Note also that, for  \mbox{$\lambda_{1}^{0} =$} $1$, the second
condition \mbox{$\lambda_{1}^{0} + \lambda_{2}^{0} - \lambda_{3}^{0} \leq$} $1$ is equivalent 
to \mbox{$\lambda_{2}^{0} =$} \mbox{$\lambda_{3}^{0}$} as noticed in Theorem~\ref{the:nosol} 
and that, in this case, the 
solution $y$ of Theorem~\ref{the:existence} is identical to the solution $x$ of 
Theorem~\ref{the:nosol}. 
At this point it is natural to ask whether the equations above have a solution iff 
\mbox{$\lambda_{1}^{0} + \lambda_{2}^{0} - \lambda_{3}^{0} \leq$} $1$. 
No answer is available at the moment.

\subsection{A property of equivalence under unitary operations} 
\label{sec:equivalence}
\begin{theorem} \label{the:equivalence}
Suppose $\sigma$ and $\tau$ are mixed states of $A \otimes B$ such that:
\begin{enumerate}
\item \mbox{$Tr_{B}( \tau ) =$} \mbox{$Tr_{B}( \sigma )$},
\mbox{$Tr_{A}( \tau ) =$} \mbox{$Tr_{A}( \sigma )$} and no eigenvalue
of any of those operators is degenerate 
\item and there are unitary maps \mbox{$u_{A} : A \rightarrow A$} and 
\mbox{$u_{B} : B \rightarrow B$} such that \mbox{$\tau =$} 
\mbox{$( u_{A} \otimes u_{B} ) \circ \sigma \circ ( u_{A}^{\ast} \otimes u_{B}^{\ast} )$}.
\end{enumerate}
Then, if $\sigma_{i , j}$ and $\tau_{i , j}$ are the elements 
of the matrices representing $\sigma$ and $\tau$ in the basis whose elements are
the tensor products of eigenbases for the traces on $A$ and $B$, then, for any
$i , j$ one has: \mbox{$\mid \tau_{i , j} \mid =$} \mbox{$\mid \sigma_{i , j} \mid$}.
\end{theorem}
\begin{proof}
Assume all assumptions of the theorem are satisfied.
By Theorem~\ref{the:aux}, we have:
\[
Tr_{B} ( \sigma ) \ = \ Tr_{B} ( \tau ) \ = \ 
Tr_{B} ( (u_{A} \otimes u_{B} ) \circ \sigma \circ ( u_{A}^{\ast} \otimes u_{B}^{\ast}) ) 
\ = \ 
u_{A} \circ Tr_{B} ( \sigma ) \circ u_{A}^{\ast}.
\]
Therefore we have: 
\mbox{$Tr_{B} ( \sigma ) \circ u_{A} \ = \ u_{A} \circ Tr_{B} ( \sigma )$}.
Let $x$ be an eigenvector of $Tr_{B}( \sigma )$ for eigenvalue $\lambda$.
We have:
\[
\lambda \, u_{A} ( x ) \ = \ u_{A} ( \lambda \, x ) \ = \ u_{A} ( Tr_{B} ( \sigma ) ( x ))
\ = \ Tr_{B} ( \sigma ) ( u_{A} ( x )).
\]
We see that $u_{A}(x)$ is an eigenvector of $Tr_{B}( \sigma )$ 
for eigenvalue $\lambda$. Since $\lambda$ is a non-degenerate eigenvalue of
$Tr_{B} ( \sigma )$, $u_{A} ( x )$ is colinear with $x$ and
\mbox{$u_{A} ( x) =$} \mbox{$e^{i \varphi} x$} for some 
\mbox{$\varphi \in [0 , 2 \pi [$}.
We see that, in a basis of eigenvectors of $Tr_{B} ( \sigma )$, the unitary operation
$u_{A}$ is represented by a diagonal matrix 
(whose diagonal entries all have modulus $1$).
Similarly for $u_{B}$ in a basis of eigenvectors of $Tr_{A} ( \sigma )$.
We conclude that, in the basis whose elements are the tensor products of the bases
of $A$ and $B$ just considered, the unitary $u_{A} \otimes u_{B}$ is represented
by a diagonal matrix (whose diagonal entries have modulus $1$).
The conclusion of the theorem follows.
\end{proof}

\subsection{Mixed state equivalence for $2$ qubits} \label{sec:mixed-equiv}
We can now give an example of how the generalization of the {\em if} part 
of Nielsen's result to mixed states fail, even for $2$-entanglement, as announced
in Section~\ref{sec:partial-mixed}.
\begin{theorem} \label{the:no-mixed-2}
Assume \mbox{$\lambda_{1}^{0} + \lambda_{2}^{0} + \lambda_{3}^{0} \leq 2$} and
\mbox{$\lambda_{3}^{0} > 1 / 2$}.
Let $y$ and $z$ be the solutions of Theorem~\ref{the:existence}, 
which both exist under the assumptions.
Let $\sigma$ be the two-qubits mixed state defined by 
\mbox{$\sigma =$} \mbox{$Tr_{\bQ_{3}} ( p_{y})$} and let
\mbox{$\tau =$} \mbox{$Tr_{\bQ_{3}} (p_{z})$}.
The spectra of $\sigma$ and $\tau$ are equal, their partial traces are the same:
\mbox{$Tr_{\bQ_{2}} ( \sigma ) =$} \mbox{$Tr_{\bQ_{2}} ( \tau )$} and
\mbox{$Tr_{\bQ_{1}} ( \sigma ) =$} \mbox{$Tr_{\bQ_{1}} ( \tau )$}, but
$\sigma$ and $\tau$ are not equivalent under local unitary transformations.
\end{theorem}
\begin{proof}
By properties of $2$-entanglement, 
the spectrum of $\sigma$ is equal (with suitable padding with zeros) 
to the spectrum of \mbox{$Tr_{\bQ_{1} \otimes \bQ_{2}} (p_{y})$}.
The spectrum of $\tau$ is equal 
to the spectrum of \mbox{$Tr_{\bQ_{1} \otimes \bQ_{2}} (p_{z})$} and
those traces are equal by construction.
By Equation~\ref{eq:three}, \mbox{$Tr_{\bQ_{2}} ( \sigma ) =$}
\mbox{$Tr_{\bQ_{2} \otimes \bQ_{3}} ( p_{y} ) =$}
\mbox{$Tr_{\bQ_{2} \otimes \bQ_{3}} ( p_{z} ) =$}
\mbox{$Tr_{\bQ_{2}} ( \tau )$} and similarly for 
\mbox{$Tr_{\bQ_{1}} ( \sigma )$}.
We are left to prove that $\sigma$ and $\tau$ are not equivalent under local unitary
transformations. We shall use the contrapositive of Theorem~\ref{the:equivalence}.
The conclusion of the theorem does not hold since, for example, we have
\mbox{$\sigma_{1 , 1} =$} 
\mbox{$y_{000} y_{000}^{\ast} + y_{001} y_{001}^{\ast} =$} 
\mbox{$( \lambda_{1}^{0} + \lambda_{2}^{0} + \lambda_{3}^{0} - 1 ) \: / \: 2$}
and \mbox{$\tau_{1 , 1} =$} 
\mbox{$z_{000} z_{000}^{\ast} + z_{001} z_{001}^{\ast} =$} 
\mbox{$( \lambda_{1}^{0} + \lambda_{2}^{0} - \lambda_{3}^{0} ) \: / \: 2$}.
Both are positive real numbers and they are different since 
\mbox{$\lambda_{3}^{0} > 1 / 2$}. But the first assumption holds:
note no eigenvalue is degenerate since  \mbox{$\lambda_{3}^{0} > 1 / 2$} implies
\mbox{$\lambda_{i}^{0} > \lambda_{i}^{1}$} for every $i$. 
Since $x$ and $y$ are solutions of the equations we know that 
\mbox{$\rho_{i}^{y} =$} $\rho_{i}^{x}$ (for every $i$).
We conclude that the second hypothesis does not hold.
\end{proof}

\subsection{Pure state equivalence under local unitary operations} \label{sec:equiv-unitary}
We are now interested in studying whether any two solutions of 
Equations~(\ref{eq:unit})
to (\ref{eq:zero-eigen2}) above
are equivalent under local unitary operations.

Two solutions that differ only by the phase factors $\theta$ are equivalent
but, at least in the generic situation described in Theorem~\ref{the:no-mixed-2}, 
the solutions $y$ and $z$ are not equivalent.

For $2$-party entanglement, pure global states $x$ and $y$ of $A \otimes B$ are equivalent
under local unitary operations iff the mixed states \mbox{$\rho_{A}^{x} =$} 
\mbox{$Tr_{B}( p_{x} )$} and \mbox{$\rho_{A}^{y} =$} 
\mbox{$Tr_{B}( p_{y} )$} have the same spectrum ($p_{z}$ is the projection on $z$).
This result cannot be generalized to $3$-party entanglement.
\begin{theorem} \label{the:3qubits-negative}
Assume \mbox{$\lambda_{1}^{0} + \lambda_{2}^{0} + \lambda_{3}^{0} \leq 2$} and
\mbox{$\lambda_{3}^{0} > 1 / 2$}.
Let $y$ and $z$ be the solutions of Theorem~\ref{the:existence}, 
which both exist under the assumptions.
They define the same spectra on each of the $\bQ_{i}$ for \mbox{$i = 1 , 2 , 3$} 
but they are not
equivalent under local unitary transformations. 
\end{theorem}
\begin{proof}
By Theorem~\ref{the:no-mixed-2}.
\end{proof}

The question whether pure states $y$ and $z$ as above can be obtained 
from each other with probability $1$ by LOCC operations is open.

\section*{Acknowledgements}
Dorit Aharonov convinced me that entanglement was fundamental and pointed
me to the right places. Discussions with her and with Michael Ben-Or 
are gratefully acknowledged.

\bibliographystyle{plain}

\appendix
\section{A theorem of Y. Fan} \label{sec:Fan}
\begin{definition} \label{def:k-basis}
Let $\bH$ be an $n$-dimensional Hilbert space and let $k$ be a natural number 
\mbox{$1 \leq k \leq n$}. A set \mbox{$x_{1} , \ldots , x_{k}$} of vectors of $\bH$ is said to
be a $k$-basis iff all the vectors $x_{i}$ are unit vectors and pairwise orthogonal.
\end{definition}
The following is a slightly modified version of a result (Theorem 1) of Y. Fan~\cite{Fan:49}.
The proof presented here is Fan's.
\begin{theorem}[Y. Fan , 1949] \label{the:Fan}
Let $\bH$ be an $n$-dimensional Hilbert space and let $q$ be a natural number 
\mbox{$1 \leq q \leq n$}. Assume \mbox{$A : \bH \rightarrow \bH$} is a self-adjoint linear operator
and that \mbox{$Sp(A) =$} \mbox{$\lambda_{1} , \ldots , \lambda_{n}$} with
\mbox{$\lambda_{i} \geq \lambda_{i + 1}$} for any $i$, \mbox{$1 \leq i < n $}.
For any $q$-basis \mbox{$b =$}
\mbox{$\{ x_{i} \}_{i = 1 , \ldots , q}$}, define the real nonnegative quantity 
\mbox{$w_{A}(b) =$}
\mbox{$\sum_{i = 1}^{q} \langle x_{i} \mid A \mid x_{i} \rangle$}.
The quantity $w_{A}(b)$ is maximal (among all $q$-bases) if and only if $b$ is composed of 
eigenvectors of $A$ for each of the eigenvalues \mbox{$\lambda_{1} , \ldots , \lambda_{q}$}.
The maximal value of $w_{A}(b)$ is the sum of the $q$ largest eigenvalues of $A$: 
\mbox{$\sum_{i = 1}^{q} \lambda_{i}$}. 
\end{theorem}
\begin{proof}
Let \mbox{$\{ x_{i} \}_{i = 1 , \ldots , q}$} be any $q$-basis and let
\mbox{$\{ y_{j} \}_{j = 1 , \ldots , n}$}, be a basis for $\bH$ with $y_{j}$ an eigenvector of $A$ for
eigenvalue $\lambda_{j}$, for \mbox{$j = 1 , \ldots , n$}. We have:
\[
\langle x_{i} \mid A \mid x_{i} \ \rangle = \ 
\sum_{j = 1}^{n} \lambda_{j} \mid \langle y_{j} \mid x_{i} \rangle \mid^{2} \ = \
\]
\[
\sum_{j = 1}^{n}  ( \lambda_{j} - \lambda_{q} )
\mid \langle y_{j} \mid x_{i} \rangle \mid^{2}  + 
\lambda_{q} \sum_{j = 1}^{n} \mid \langle y_{j} \mid x_{i} \rangle \mid^{2} \ = \ 
\]
\[
\sum_{j = 1}^{n}  ( \lambda_{j} - \lambda_{q} )
\mid \langle y_{j} \mid x_{i} \rangle \mid^{2}  + 
\lambda_{q} \| x_{i} \|^{2} \ = \ 
\sum_{j = 1}^{n}  ( \lambda_{j} - \lambda_{q} )
\mid \langle y_{j} \mid x_{i} \rangle \mid^{2}  + \lambda_{q} 
\]
Therefore
\[
\sum_{i = 1}^{q} \langle x_{i} \mid A \mid x_{i} \ \rangle \ = \ 
\sum_{j = 1}^{n}  ( \lambda_{j} - \lambda_{q} ) \sum_{i = 1}^{q}  
\mid \langle y_{j} \mid x_{i} \rangle \mid^{2}  + q \lambda_{q}.
\]
Since \mbox{$\lambda_{j} - \lambda_{q} \leq 0$} for \mbox{$j \geq q$} and
\mbox{$\lambda_{j} - \lambda_{q} \geq 0$} for \mbox{$j \leq q$}
we have:
\[
\sum_{i = 1}^{q} \langle x_{i} \mid A \mid x_{i} \ \rangle \ \leq \ 
\sum_{j = 1}^{q}  ( \lambda_{j} - \lambda_{q} ) \| y_{j} \|^{2}  + q \lambda_{q} \ = \
\sum_{j = 1}^{q}  ( \lambda_{j} - \lambda_{q} ) + q \lambda_{q} \ = \ \sum_{j = 1}^{q}  \lambda_{j}
\]
and \mbox{$\sum_{j = 1}^{q}  \lambda_{j}$} is an upper bound for $w_{A}(b)$.
But we have \mbox{$\sum_{i = 1}^{q} \langle x_{i} \mid A \mid x_{i} \ \rangle =$}
\mbox{$\sum_{j = 1}^{q}  \lambda_{j}$}
if and only if, 
\begin{itemize}
\item for any $j$, \mbox{$q < j \leq n$} one has
\mbox{$( \lambda_{j} - \lambda_{q} ) 
\sum_{i = 1}^{q} \mid \langle y_{j} \mid x_{i} \rangle \mid^{2} =$} $0$, i.e., for any $j$
such that \mbox{$\lambda_{j} < \lambda_{q}$} and for any $i$, \mbox{$1 \leq i \leq q$}
the vectors $y_{j}$ and $x_{i}$ are orthogonal, and
\item for any $j$, \mbox{$1 \leq j \leq q$}, the vector $y_{j}$ is orthogonal to all vectors $x_{i}$ 
for \mbox{$q < i \leq n$}.
\end{itemize}
We conclude that $w_{A}(b)$ is equal to \mbox{$\sum_{j = 1}^{q}  \lambda_{j}$} iff, for every
$q$, \mbox{$1 \leq q \leq n$}, the subspace spanned by \mbox{$\{ x_{i} \}_{i = 1 , \ldots , q}$}
is the subspace spanned by \mbox{$\{ y_{i} \}_{i = 1 , \ldots , q}$}.
\end{proof}

\section{The state resulting of a measurement majorizes the initial state only in expectation}
\label{sec:expect-only}
Consider a global space \mbox{$\bH =$} \mbox{$\bH_{a} \otimes \bH_{b}$} where $\bH_{a}$, 
Alice's system, consists of two qubits (qubits $1$ and $2$) and Bob's system consists of one qubit
(qubit $3$). Let the global state be 
\[
h \ = \ \sqrt{p / 3} \, \mid 000 \rangle \, + \, \sqrt{p / 3} \, \mid 001 \rangle \, + \,  
\sqrt{p / 3} \, \mid 111 \rangle \, + \, 
\]
\[
\sqrt{( 1 - p ) / 2} \, \mid 010 \rangle \, + \, \sqrt{ (1 - p ) / 2} \, \mid 100 \rangle
\]
for some $p$, \mbox{$0 \leq p < 1$}.
The local mixed state for Bob is 
\[
\left( \begin{array}{cc}
1 - 2 p / 3 & p / 3 \\ p / 3 & 2 p / 3
\end{array} \right)
\]
whose eigenvalues are \mbox{$(1\pm \sqrt{1 - 4 / 3 \, p (2 - 5 / 3 \, p)} )  / 2$}.
As expected, if $p$ is small, the mixed state for Bob is almost the pure state $\mid 0 \rangle$,
and one of the eigenvalues is close to $1$, the other close to $0$.
If Alice's measures her local state by testing it on the orthogonal subspaces
spanned by \mbox{$\mid 00 \rangle$} and \mbox{$\mid 11 \rangle$} 
on one hand and by \mbox{$\mid 01 \rangle$} and \mbox{$\mid 10 \rangle$} 
on the other hand and if she gets the first subspace as an answer, 
then the global state of the system will be:
\[
h' \ = \ \sqrt{1 / 3} (\mid 000 \rangle \, + \, \mid 001 \rangle \, + \, \mid 111 \rangle )
\]
and Bob's local mixed state will be
\[
\left( \begin{array}{cc}
1 / 3 & 1 / 3 \\ 1 / 3 & 2 / 3
\end{array} \right)
\]
whose eigenvalues are \mbox{$(1 \pm \sqrt{ 5 } / 3 ) / 2$}.
If \mbox{$p < 0.2$} then 
\mbox{$ ( 1 + \sqrt{ 5} / 3 ) / 2$} is less than 
\mbox{$( 1 + \sqrt{1 - 4 / 3 \, p ( 2 - 5 / 3 \, p)} ) / 2$} and
the mixed state of Bob after Alice's measurement is strictly more mixed, 
i.e., less pure than it was before.

\section{Something probably well-known} \label{sec:well-known}
I guess the following is well-known but I miss a precise reference.
\begin{theorem} \label{the:sp=}
Let $A$ be a finite dimensional Hilbert space and \mbox{$f : A \rightarrow A$} a linear operator.
Then, \mbox{$Sp ({f}^{\ast} \circ f ) =$}
\mbox{$Sp ( f \circ {f}^{\ast} )$}.
\end{theorem}
\begin{proof}
Note, first, that both ${f}^{\ast} \circ f$ and $f \circ {f}^{\ast}$ are self-adjoint
and therefore have $dim(A)$ real eigenvalues.
We shall show that every eigenvalue $\lambda$ of ${f}^{\ast} \circ f$, 
different from zero,
is an eigenvalue of $f \circ {f}^{\ast}$ with the same multiplicity.
To this effect we note that if \mbox{$x \in A$} is an eigenvector of ${f}^{\ast} \circ f$ 
for some eigenvalue \mbox{$\lambda \neq 0$}, then \mbox{$f ( x )$} is an eigenvector of 
$f \circ {f}^{\ast}$ for eigenvalue $\lambda$.
Suppose indeed $x$ and $\lambda$ are as assumed, then
\mbox{${f}^{\ast} ( f (x )) =$} \mbox{$\lambda \, x \neq \vec{0}$} and therefore
\mbox{$f ( x ) \neq \vec{0}$}. 
But \mbox{$(f \circ {f}^{\ast}) (f (x ) ) =$}
\mbox{$f ( ( {f}^{\ast} \circ f ) ( x) =$}
\mbox{$f ( \lambda \, x ) =$}
\mbox{$\lambda \, f ( x )$}.
We are left to show that the multiplicity of $\lambda$ for $f \circ {f}^{\ast}$ is at least its
multiplicity for ${f}^{\ast} \circ f$.
For this, we note that if \mbox{$y \in A$} is orthogonal to $x$, 
then $f ( y )$ is orthogonal to $f ( x )$.
Indeed, \mbox{$\langle f ( y ) \mid f ( x ) \rangle =$}
\mbox{$\langle y \mid ( {f}^{\ast} \circ f ) ( x ) \rangle =$}
\mbox{$ \langle y \mid \lambda \, x \rangle =$} 
\mbox{$\lambda \, \langle y \mid x \rangle = 0$}.
\end{proof}

\end{document}